\documentclass{birkmult}
\usepackage{amssymb}
\newtheorem{thm}{Theorem}[section]
 \newtheorem{cor}[thm]{Corollary}
 
 \newtheorem{prop}[thm]{Proposition}
 \theoremstyle{definition}
 \newtheorem{ddefn}[thm]{Definition}
\renewcommand{\theequation}{\thesection.\arabic{equation}}
 \theoremstyle{definition}
 
 \theoremstyle{remark}
 
\begin{document}
\title[On pseudo-Hermitian operators with generalized $\mathcal{C}$-symmetries]{On pseudo-Hermitian operators with \\ generalized $\mathcal{C}$-symmetries}
\author{S. Kuzhel}
\address{Institute of Mathematics \\
National Academy of Sciences of Ukraine \\
Tereshchenkivs'ka 3 \\
01601 Kiev  \\
Ukraine}
\email{kuzhel@imath.kiev.ua}

\begin{abstract} The concept of $\mathcal{C}$-symmetries for
pseudo-Hermitian Hamiltonians is studied in the Krein space
framework. A generalization of $\mathcal{C}$-symmetries is
suggested.
\end{abstract}
 \maketitle

\section{Introduction}
The employing of non-Hermitian operators for the description of
experimentally observable data goes back to the early days of
quantum mechanics \cite{Di, Gu, Pa}. A steady interest to
non-Hermitian Hamiltonians became enormous after it has been
discovered that complex Hamiltonians possessing so-called
$\mathcal{PT}$-symmetry (the product of parity and time reversal)
can have a real spectrum (like self-adjoint operators) \cite{B1, D2,
Z3, Z2}. The results obtained gave rise to a consistent complex
extension of the standard quantum mechanics (see the review paper
\cite{B4} and the references therein).

One of key moments in the $\mathcal{PT}$-symmetric quantum theory is
the description of a previously unnoticed symmetry (hidden symmetry)
for a given ${\mathcal P}{\mathcal T}$-symmetric Hamiltonian $A$
that is represented by a linear operator $\mathcal{C}$. The
properties of $\mathcal{C}$ are nearly identical to those of the
charge conjugation operator in quantum field theory and the
existence of $\mathcal{C}$ provides an inner product whose
associated norm is positive definite and the dynamics generated by
$A$ is then governed by a unitary time evolution. However, the
operator $\mathcal{C}$ depends on the choice of $A$ and its finding
is a nontrivial problem \cite{B, B22, BJ}.

The concept of ${\mathcal P}{\mathcal T}$-symmetry can be placed in
a more general mathematical context known as
\emph{pseudo-Hermiticity}. A linear densely defined operator $A$
acting in a Hilbert space ${\mathfrak H}$  is pseudo-Hermitian if
there is an invertible bounded self-adjoint operator $\eta :
{\mathfrak H} \to {\mathfrak H}$ such that
\begin{equation}\label{e1}
{A}^*\eta=\eta{A},
\end{equation}
where the sign $*$ stands for the adjoint of the corresponding
operator. Including the concept of ${\mathcal P}{\mathcal
T}$-symmetry in a pseudo-Hermitian framework enables one to make
more clear basic constructions of ${\mathcal P}{\mathcal
T}$-symmetric quantum mechanics and to achieve a lot of nontrivial
physical results \cite{AX, BM, GF, JO, M1, M3}.

The related notion of quasi-Hermiticity and its physical
implications were discussed in detail in \cite{GHS, SGH}.

Using Langer's observation \cite{L1} that a Hilbert space
$\mathfrak{H}$ with the indefinite metric $[f,g]_{\eta}=(\eta{f},g)$
($0\in\rho(\eta)$) is a Krein space, one can reduce the
investigation of pseudo-Hermitian operators to the study of
self-adjoint operators in a Krein space \cite{A3, AK2, LT, MO4, TT}.
The present paper continues such trend of investigations and its aim
is to analyze pseudo-Hermitian operators with
$\mathcal{C}$-symmetries in the Krein space setting. The special
attention will be paid to the generalization of the concept of
$\mathcal{C}$ operators.

The existence of a $\mathcal{C}$-symmetry for a pseudo-Hermitian
operator $A$ means that $A$ has a maximal dual pair
$\{\mathfrak{L}_+, \mathfrak{L}_-\}$ of invariant subspaces of
$\mathfrak{H}$ \cite{HL, HL1}.

The paper is organized as follows. Section 2 contains all necessary
Krein space results in the form convenient for our presentation.
Their proofs and detailed analysis can be found in \cite{AZ}.
Section 3 deals with the study of $\mathcal{C}$-symmetries by the
Krein space methods. A more physical presentation of the subject can
be found in \cite{B} -- \cite{BT, JO, MO4, TT}. Section 4 contains
some generalization of the concept of $\mathcal{C}$-symmetry where
operators $\mathcal{C}$ are supposed to be unbounded in
$\mathfrak{H}$. The case of unbounded metric operator $\eta$ has
been recently studied in \cite{KS}. Examples of
$\mathcal{C}$-symmetries and generalized $\mathcal{C}$-symmetries
are presented in Section 5.

For the sake of simplicity, we restrict ourselves to the case where
a self-adjoint operator $\eta$ is also unitary. For such a
self-adjoint and unitary operator $\eta$ the notation $J$ (i.e.
$\eta\equiv{J}$) will be used. Note that the requirement of
unitarity of $\eta$ is not restrictive because the general case of a
self-adjoint operator $\eta$ is reduced to the case above if,
instead of the original scalar product $(\cdot, \cdot)$ in
$\mathfrak{H}$, ones consider another (equivalent to it) scalar
product $(|\eta|\cdot,\cdot)$, where $|\eta|=\sqrt{\eta^2}$ is the
modulus of $\eta$.

\section{Elements of the Krein's spaces theory}

Let ${\mathfrak H}$ be a Hilbert space with scalar product
$(\cdot,\cdot)$ and let $J$ be a fundamental symmetry in ${\mathfrak
H}$  (i.e., $J=J^*$ and  $J^2=I$). The corresponding orthoprojectors
$P_+=1/2(I+J)$, $P_-=1/2(I-J)$ determine the fundamental
decomposition of ${\mathfrak H}$
\begin{equation}\label{d1}
{\mathfrak H}={\mathfrak H}_+\oplus{\mathfrak H}_-, \qquad
{\mathfrak H}_-=P_{-}{\mathfrak H}, \quad {\mathfrak
H}_+=P_{+}{\mathfrak H}.
\end{equation}

The space ${\mathfrak H}$ with an indefinite scalar product
(indefinite metric)
\begin{equation}\label{e29}
 [x,y]:=(J{x}, y), \qquad \forall{x,y}\in\mathfrak{H}
\end{equation}
is called {\it a Krein space} if $\dim{\mathfrak H}_+=\dim{\mathfrak
H}_-=\infty$.

A (closed) subspace $\mathfrak{L}\subset{\mathfrak H}$ is called
{\it nonnegative, positive, uniformly positive} if, respectively, $
[x,x]\geq0, \quad [x,x]>0, \quad [x,x]\geq\alpha\|x\|^2$ for all
$x\in\mathfrak{L}\setminus\{0\}$. Nonpositive, negative, and
uniformly negative subspaces are introduced similarly. The subspaces
$\mathfrak{H}_+$ and $\mathfrak{H}_-$ in (\ref{d1}) are maximal
uniformly positive and maximal uniformly negative, respectively.

Let $\mathfrak{L}_+$ be a maximal positive subspace (i.e., the
closed set $\mathfrak{L}_+$ does not belong as a subspace to any
positive subspace). Its $J$-orthogonal complement
$$
\mathfrak{L}_-=\mathfrak{L}_+^{[\bot]}=\{x\in\mathfrak{H} \ | \
[x,y]=0, \ \forall{y}\in\mathfrak{L}_+\}
$$ is a maximal negative
subspace of $\mathfrak{H}$ and the direct $J$-orthogonal sum
\begin{equation}\label{d2}
{\mathfrak H}'={\mathfrak L}_+[\dot{+}]{\mathfrak L}_-
\end{equation}
is a dense set in $\mathfrak{H}$. (Here the brackets $[\cdot]$ means
the orthogonality with respect to the indefinite metric.) The linear
set ${\mathfrak H}'$ coincides with $\mathfrak{H}$ if and only if
$\mathfrak{L}_+$ is a maximal uniformly positive subspace. In that
case ${\mathfrak L}_-=\mathfrak{L}_+^{[\bot]}$ is a maximal
uniformly negative subspace.

The subspaces ${\mathfrak L}_+$ and ${\mathfrak L}_-$ in (\ref{d2})
can be decomposed as follows:
$$
{\mathfrak L}_+=(I+K)\mathfrak{H}_+ \qquad {\mathfrak
L}_-=(I+Q)\mathfrak{H}_-,
$$
where $K:\mathfrak{H}_+\to\mathfrak{H}_-$ is a contraction and
($Q=K^*:\mathfrak{H}_-\to\mathfrak{H}_+$) coincides with the adjoint
of $K$.

The self-adjoint operator $T=KP_++K^*P_-$ acting in $\mathfrak{H}$
is called {\it an operator of transition} from the fundamental
decomposition (\ref{d1}) to (\ref{d2}). Obviously,
\begin{equation}\label{d4}
{\mathfrak L}_+=(I+T)\mathfrak{H}_+, \qquad {\mathfrak
L}_-=(I+T)\mathfrak{H}_-.
\end{equation}

The collection of operators of transition admits a simple `external'
description. Namely, a self-adjoint operator $T$ in  $\mathfrak{H}$
is an operator of transition if and only if
\begin{equation}\label{e3}
\|Tx\|<\|x\| \quad (\forall{x\not=0}) \qquad \mbox{and} \qquad
JT=-TJ.
\end{equation}
The important particular case of (\ref{d2}), where $\mathfrak{H}'$
coincides with $\mathfrak{H}$ (i.e., $\mathfrak{H}={\mathfrak
L}_+[\dot{+}]{\mathfrak L}_-$) corresponds to the more strong
condition $\|T\|<1$ in (\ref{e3}).

Let $P_{\mathfrak{L}_\pm} : \mathfrak{H}'\to{\mathcal L}_{\pm}$ be
the projectors onto $\mathfrak{L}_\pm$ with respect to decomposition
(\ref{d2}). Repeating step by step the proof of Proposition 9.1 in
\cite{KK} where the case $\mathfrak{H}={\mathfrak
L}_+[\dot{+}]{\mathfrak L}_-$ has been considered, one gets
\begin{equation}\label{k1}
P_{\mathfrak{L}_-}=(I-T)^{-1}(P_--TP_+),  \qquad
P_{\mathfrak{L}_+}=(I-T)^{-1}(P_+-TP_-),
\end{equation}
where $T$ is the operator of transition from \eqref{d1} to
\eqref{d2}.

\section{The Condition of $\mathcal{C}$-symmetry}

A linear densely defined operator $A$ acting in a Krein space
$({\mathfrak H}, [\cdot, \cdot])$ is called $J$-{\it self-adjoint}
if its adjoint ${A}^*$ satisfies the condition ${A}^*J=J{A}$.
Obviously, $J$-self-adjoint operators are pseudo-Hermitian ones in
the sense of \eqref{e1}.

Since a $J$-self-adjoint operator ${A}$ is self-adjoint with respect
to the indefinite metric \eqref{e29}, one can attempt to develop a
consistent quantum theory for $J$-self-adjoint Hamiltonians with
real spectrum. However, in this case, we encounter the difficulty of
dealing with the indefinite metric $[\cdot,\cdot]$. Since the norm
of states carries a probabilistic
 interpretation in the standard quantum theory, the presence of an
 indefinite metric immediately raises problems of interpretation.
 One of the natural ways to overcome this problem consists in the
 construction of a hidden symmetry of $A$ that is represented
 by the linear operator $\mathcal{C}$. This symmetry operator $\mathcal{C}$
 guarantees that the pseudo-Hermitian Hamiltonian $A$ can be used
 to define a unitary theory of quantum mechanics \cite{B4, BBJ}.

By analogy with \cite{B4} the definition of $\mathcal{C}$ can be
formalized as follows.
\begin{ddefn}\label{dad1}
A $J$-self-adjoint operator ${A}$ has the property of
$\mathcal{C}$-symmetry if there exists a bounded linear operator
$\mathcal{C}$ in $\mathfrak{H}$ such that: \ $(i) \
{\mathcal{C}}^2=I;$ \quad $(ii) \ J\mathcal{C}>0$; \quad $(iii) \
A{\mathcal{C}}={\mathcal{C}}A$.
\end{ddefn}

The next simple statement clarifies the structure of
pseudo-Hermitian operators with $\mathcal{C}$-symmetries.

\begin{prop}\label{t1}
Let $A$ be a $J$-self-adjoint operator. Then $A$ has the property of
${\mathcal C}$-symmetry if and only if $A$ admits the decomposition
\begin{equation}\label{new1}
A=A_+[\dot{+}]A_-, \qquad A_{+}=A\upharpoonright\mathfrak{L}_{+},
\quad A_{-}=A\upharpoonright\mathfrak{L}_{-}
\end{equation}
with respect to a certain choice of $J$-orthogonal decomposition of
$\mathfrak{H}$
\begin{equation}\label{e25}
 \mathfrak{H}=\mathfrak{L}_+[\dot{+}]\mathfrak{L}_-, \quad
 \mathfrak{L}_-=\mathfrak{L}_+^{[\bot]},
\end{equation}
where $\mathfrak{L}_+$ is a maximal uniformly positive subspace of
the Krein space $(\mathfrak{H}, [\cdot,\cdot])$.
\end{prop}

\begin{proof}
Let $A$ admit the decomposition \eqref{new1} with respect to
\eqref{e25}. Denote
$\mathcal{C}=P_{\mathfrak{L}_+}-P_{\mathfrak{L}_-}$, where
$P_{\mathfrak{L}_\pm}$ are projectors onto $\mathfrak{L}_\pm$
according to \eqref{e25}. Obviously, the bounded linear operator
$\mathcal{C}$ satisfies $\mathcal{C}^2=I$ and
$\mathcal{C}A=A\mathcal{C}$. Furthermore, by virtue of the second
relation in \eqref{e3} and \eqref{k1},
\begin{equation}\label{ee12}
\mathcal{C}=P_{\mathfrak{L}_+}-P_{\mathfrak{L}_-}=(I-T)^{-1}(I+T)(P_+-P_-)=J(I-T)(I+T)^{-1},
\end{equation}
where $T$ is the operator of transition from \eqref{d1} to
\eqref{e25}. Hence $J\mathcal{C}=(I-T)(I+T)^{-1}>0$  (since
$\|T\|<1$). Thus $A$ has $\mathcal{C}$-symmetry.

Conversely, assume that $A$ has $\mathcal{C}$-symmetry and denote
${\mathfrak L}_{+}=(I+\mathcal{C})\mathfrak{H}$ and ${\mathfrak
L}_{-}=(I-\mathcal{C})\mathfrak{H}$. Since $\mathcal{C}^2=I$, one
gets $\mathfrak{H}={\mathfrak L}_+\dot{+}{\mathfrak L}_-$ and
\begin{equation}\label{new24}
\mathcal{C}f_{-}=-{f}_{-}, \quad \mathcal{C}f_{+}={f}_{+}, \qquad
\forall{f_{\pm}}\in{\mathfrak L}_{\pm}.
\end{equation}
Therefore,
$$
[f_+,f_+]=[\mathcal{C}f_+,f_+]=(J\mathcal{C}f_+,f_+)>0, \quad
[f_-,f_-]=-[\mathcal{C}f_-,f_-]=-(J\mathcal{C}f_-,f_-)<0.
$$
Thus ${\mathfrak L}_{+}$ (${\mathfrak L}_{-}$) is a positive
(negative) linear set of $\mathfrak{H}$.

The property of $J\mathcal{C}$ to be self-adjoint in $\mathfrak{H}$
implies that $\mathcal{C}^*J=J\mathcal{C}$, i.e., $\mathcal{C}$ is
J-self-adjoint. In that case
$$
[f_+,f_-]=[\mathcal{C}f_+,f_-]=[f_+,\mathcal{C}f_-]=-[f_+,f_-]
$$
and hence, $[f_+,f_-]=0$.

Summing the results established above, one concludes that the
operator $\mathcal{C}$ determines an $J$-orthogonal decomposition
\eqref{e25} of $\mathfrak{H}$, where ${\mathfrak L}_{+}$ and
${\mathfrak L}_{-}$ are positive and negative linear subspaces of
$\mathfrak{H}$. Such type of decomposition is possible only in the
case where $\mathfrak{L}_+$ is a maximal uniformly  positive
subspace of $\mathfrak{H}$ and
$\mathfrak{L}_-=\mathfrak{L}_+^{[\bot]}$ \cite{AZ}.

To complete the proof it suffices to observe that \eqref{new1}
follows from the relations $A\mathcal{C}=\mathcal{C}A$ and
${\mathfrak L}_{\pm}=(I\pm\mathcal{C})\mathfrak{H}$.
\end{proof}

\begin{cor}\label{c1}
A $J$-self-adjoint operator $A$ has the property of ${\mathcal
C}$-symmetry if and only if
$H=\sqrt{J{\mathcal{C}}}A(\sqrt{J{\mathcal{C}}})^{-1}$ is a
self-adjoint operator in $\mathfrak{H}$.
\end{cor}

\begin{proof}
Denote for brevity $F=J{\mathcal{C}}$. It follows from the
conditions $(i), (ii)$ of the definition of $\mathcal{C}$-symmetry
that $F$ is a bounded uniformly positive operator in $\mathfrak{H}$.

If a $J$-self-adjoint operator $A$ has the property of ${\mathcal
C}$-symmetry, then
$$
(\sqrt{F}Ax,\sqrt{F}y)=[\mathcal{C}Ax, y]=[A\mathcal{C}x,
y]=[\mathcal{C}x, Ay]=(\sqrt{F}x,\sqrt{F}Ay).
$$
This means that $H=\sqrt{F}A(\sqrt{F})^{-1}$ is a self-adjoint
operator in $\mathfrak{H}$ with respect to the initial product
$(\cdot,\cdot)$.

Conversely, if $H=\sqrt{F}A(\sqrt{F})^{-1}$ is self-adjoint, then
$$
[\mathcal{C}Ax, y]=(H\sqrt{F}x, \sqrt{F}y)=(\sqrt{F}x,
H\sqrt{F}y)=[\mathcal{C}x, Ay]=[A\mathcal{C}x, y]
$$
for any $x,y\in\mathfrak{H}$. Therefore, $\mathcal{C}A=A\mathcal{C}$
and $A$ has ${\mathcal C}$-symmetry.
\end{proof}

It follows from the proof that the scalar product
$(x,y)_{\mathcal{C}}:=[{\mathcal{C}}x,y]$ determined by
$\mathcal{C}$ is equivalent to the initial scalar product
$(\cdot,\cdot)$ in $\mathfrak{H}$. Thus the existence of a
${\mathcal{C}}$-symmetry for a $J$-self-adjoint operator $A$ ensures
unitarity of the dynamics generated by $A$ in the norm
$\|\cdot\|^2_{\mathcal{C}}=(\cdot,\cdot)_{\mathcal{C}}$ equivalent
to the initial one.

In contrast to Proposition \ref{t1}, Corollary \ref{c1} does not
emphasize the property of $A$ to be diagonalizable into two operator
parts in $\mathfrak{H}$.

Denote
\begin{equation}\label{new64}
U=\frac{1}{2}[(I+\mathcal{C})P_++(I-\mathcal{C})P_-].
\end{equation}

It is clear that $\mathcal{C}U=UJ$. This gives $U :
\mathfrak{H}_+\to{\mathfrak{L}_{+}}$ and $U :
\mathfrak{H}_-\to{\mathfrak{L}_{-}}$, where $\mathfrak{H}_\pm$ are
subspaces of the fundamental decomposition \eqref{d1} and
$\mathfrak{L}_\pm=(I\pm{C})\mathfrak{H}$ are reducing subspaces for
$A$ in Proposition \ref{t1}. Furthermore, it follows from
\eqref{new64} that the operators
$U_{\pm}=U\upharpoonright\mathfrak{H}_{\pm}$ determine bounded
invertible mappings of $\mathfrak{H}_\pm$ onto $\mathfrak{L}_\pm$
(since the subspaces $\mathfrak{H}_+$ and $\mathfrak{L}_{+}$ are
maximal uniformly positive and $\mathfrak{H}_-$ and
$\mathfrak{L}_{-}$ are maximal uniformly negative).

By virtue of Proposition \ref{t1}, the transformation $U$ decomposes
an $J$-self-adjoint operator $A$ with $\mathcal{C}$ symmetry into
the $2\times{2}$-block form
$$
U^{-1}AU=\left(\begin{array}{cc}
U_{+}^{-1}AU_{+} & 0 \\
0 & U_{-}^{-1}AU_{-}
\end{array}\right)
$$
with respect to the fundamental decomposition \eqref{d1}. For this
reason the mapping determined by $U$ can be considered as a
generalization of the Foldy-Wouthuysen transformation well-known in
quantum mechanics (see e.g. \cite{Ta}).

\section{Generalized ${\mathcal C}$-Symmetry}

The concept of ${\mathcal{C}}$-symmetry can be weakened as follows.
\begin{ddefn}\label{dad2}
A $J$-self-adjoint operator ${A}$ has the property of generalized
$\mathcal{C}$-symmetry if there exists a linear densely defined
operator $\mathcal{C}$ in $\mathfrak{H}$ such that: \ $(i) \
{\mathcal{C}}^2=I;$ \quad $(ii)$ \ the operator $J\mathcal{C}$ is
positive self-adjoint in $\mathfrak{H}$; \quad $(iii) \
\mathcal{D}(A)\subset\mathcal{D}(\mathcal{C}) \quad \mbox{and} \quad
A{\mathcal{C}}={\mathcal{C}}A$.
\end{ddefn}

 The main difference with definition \ref{dad1} is that the
operator $\mathcal{C}$ is not assumed to be bounded.

\begin{prop}[cf. Proposition \ref{t1}]\label{tt66}
 A $J$-self-adjoint operator $A$ has the property of generalized
$\mathcal{C}$-symmetry if and only if $A$ admits the decomposition
\begin{equation}\label{newww1}
A=A_+[\dot{+}]A_-,  \qquad A_{+}=A\upharpoonright\mathfrak{L}_{+},
\quad A_{-}=A\upharpoonright\mathfrak{L}_{-}
\end{equation}
with respect to a certain choice of $J$-orthogonal sum
\begin{equation}\label{new71}
\mathfrak{H}\supset\mathfrak{H}'=\mathfrak{L}_+[\dot{+}]\mathfrak{L}_-
\qquad \mathfrak{L}_-=\mathfrak{L}_+^{[\bot]},
\end{equation}
where $\mathfrak{L}_+$ is a maximal positive subspace of the Krein
space $(\mathfrak{H}, [\cdot,\cdot])$.
\end{prop}

\begin{proof}
If $\mathfrak{L}_+$ satisfies the condition of Proposition
\ref{tt66},  then ${\mathfrak L}_+[\dot{+}]{\mathfrak L}_-$ is a
dense set in $\mathfrak{H}$ \cite{AZ}. Denote
$\mathcal{C}=P_{\mathfrak{L}_+}-P_{\mathfrak{L}_-}$, where
$P_{\mathfrak{L}_\pm}$ are projectors in ${\mathfrak
L}_+[\dot{+}]{\mathfrak L}_-$ onto $\mathfrak{L}_\pm$. It follows
from the definition of $\mathcal{C}$ and \eqref{newww1} that
$\mathcal{C}^2=I$,  $\mathcal{C}A=A\mathcal{C}$, and
$\mathcal{D}(\mathcal{C})=\mathfrak{L}_+[\dot{+}]\mathfrak{L}_-$.

The operator $\mathcal{C}$ is also defined by \eqref{ee12}, where
$T$ is the operator of transition from \eqref{d1} to
$\mathfrak{L}_+[\dot{+}]\mathfrak{L}_-$. Since $T$ satisfies
\eqref{e3}, the operator $J\mathcal{C}=(I-T)(I+T)^{-1}$ is positive
self-adjoint in $\mathfrak{H}$. Thus $A$ has a generalized
$\mathcal{C}$-symmetry.

Conversely, assume that $A$ has generalized $\mathcal{C}$-symmetry
and denote
\begin{equation}\label{as1}
T=(I-F)(I+F)^{-1}, \qquad F=J\mathcal{C}.
\end{equation}

Since $F$ is positive self-adjoint, the operator $T$ satisfies
$\|Tx\|<\|x\|$ ($x\not=0$) and
\begin{equation}\label{as2}
J\mathcal{C}=F=(I-T)(I+T)^{-1}.
\end{equation}

The conditions $(i)$ and $(ii)$ of definition \ref{dad2} imply
$JF=\mathcal{C}=F^{-1}J$. Combining this relation with \eqref{as1}
one gets $JT=-TJ$. So, the operator $T$ satisfies the conditions of
\eqref{e3}. This means that $T$ is the operator of transition from
\eqref{d1} to the direct $J$-orthogonal sum \eqref{d2} (or
\eqref{new71}), where a maximal positive subspace $\mathfrak{L}_{+}$
has the form \eqref{d4} and
$\mathfrak{L}_-=\mathfrak{L}_+^{[\bot]}$.

Since the projectors $P_{\mathfrak{L}_{\pm}}$ are defined by
\eqref{k1}, relation \eqref{as2} implies (cf. \eqref{ee12}) $
P_{\mathfrak{L}_+}-P_{\mathfrak{L}_-}=J(I-T)(I+T)^{-1}=\mathcal{C}.$
Hence, ${\mathfrak L}_{+}=(I+\mathcal{C})\mathcal{D}(\mathcal{C})$
and ${\mathfrak L}_{-}=(I-\mathcal{C})\mathcal{D}(\mathcal{C})$. In
that case, the decomposition \eqref{newww1} immediately follows from
the relation $A\mathcal{C}=\mathcal{C}A$.
\end{proof}

\begin{cor}\label{pp21}
If a $J$-self-adjoint operator $A$ possesses a generalized
$\mathcal{C}$-symmetry given by an operator $\mathcal{C}$ in the
sense of Definition \eqref{dad2}, then its adjoint $\mathcal{C}^*$
provides the property of a generalized $\mathcal{C}$-symmetry for
$A^*$.
\end{cor}
\begin{proof}
Since $A$ is $J$-self-adjoint, the relation $AJ=JA^{*}$ holds.
Therefore, if $A$ is decomposed with respect to \eqref{new71} in the
sense of Proposition \ref{tt66}, then $A^*$ has the similar
decomposition with respect to the dense subset $J{\mathfrak
L}_+[\dot{+}]J{\mathfrak L}_-$ of $\mathfrak{H}$, where $J{\mathfrak
L}_+$ is a maximal positive subspace of the Krein space
$(\mathfrak{H}, [\cdot,\cdot])$. Therefore, $A^*$ has a generalized
$\mathcal{C}$-symmetry.

The $J$-orthogonal sum \eqref{new71} is uniquely determined by the
operator $\mathcal{C}=J(I-T)(I+T)^{-1}$, where $T$ is the operator
of transition from \eqref{d1} to
$\mathfrak{L}_+[\dot{+}]\mathfrak{L}_-$. It follows from \eqref{d4}
and \eqref{e3} that $T'=-T$ is the operator of transition from
\eqref{d1} to $J\mathfrak{L}_+[\dot{+}]J\mathfrak{L}_-$. According
to the proof of Proposition \ref{tt66} and \eqref{e3}, the operator
$$
\mathcal{C}'=J(I-T')(I+T')^{-1}=J(I+T)(I-T)^{-1}=(I-T)(I+T)^{-1}J=\mathcal{C}^*
$$
provides the property of generalized $\mathcal{C}$-symmetry for
$A^*$.
\end{proof}

\begin{cor}\label{p21}
If a $J$-self-adjoint operator $A$ has a generalized
$\mathcal{C}$-symmetry, then $\mathbb{C}\setminus\mathbb{R}$ belongs
to the continuous spectrum of $A$ (i.e.,
$\sigma_c(A)\supset\mathbb{C}\setminus\mathbb{R}$).
\end{cor}

\begin{proof} Assume that a $J$-self-adjoint operator $A$ with generalized
$\mathcal{C}$-symmetry has a non-real eigenvalue
$z\in\mathbb{C}\setminus\mathbb{R}$. By Proposition \ref{tt66}, at
least one of the operators $A_{\pm}$ in \eqref{newww1} have the
eigenvalue $z$. However, this is impossible because $A_{\pm}$ are
symmetric in the pre-Hilbert spaces $\mathfrak{L}_{\pm}$ with scalar
products $[\cdot, \cdot]$ and $-[\cdot, \cdot]$, respectively.
Therefore,
$\sigma_p(A)\cap(\mathbb{C}\setminus\mathbb{R})=\emptyset$. Further
$\sigma_r(A)\cap(\mathbb{C}\setminus\mathbb{R})=\emptyset$ since
$A^*$ has a generalized $\mathcal{C}$-symmetry.

In view of \eqref{newww1} and \eqref{new71}, ${\mathcal
R}(A-zI)\subseteq\mathfrak{L}_+[\dot{+}]\mathfrak{L}_-=\mathfrak{H}'\not=\mathfrak{H}$
for any non-real $z$. Hence,
$\sigma_c(A)\supset\mathbb{C}\setminus\mathbb{R}$.
\end{proof}

Denote by $\mathfrak{H}_{\mathcal C}$ the completion of
$\mathcal{D}(\mathcal{C})$ with respect to the positive sesquilinear
form
$$
(f, g)_{\mathcal{C}}:=[\mathcal{C}f,g]=(J\mathcal{C}f,g), \qquad
\forall{f,g}\in{\mathcal{D}(\mathcal{C}}).
$$
In contrast to the case of $\mathcal{C}$-symmetry (see Section 3),
the norm $\|\cdot\|_{\mathcal{C}}^2=(\cdot,\cdot)_{\mathcal{C}}$ is
not equivalent to the initial one.

If a $J$-self-adjoint operator $A$ has a generalized
$\mathcal{C}$-symmetry, then
$$
(Af, g)_{\mathcal{C}}=[\mathcal{C}Af, g]=[A\mathcal{C}f,
g]=[\mathcal{C}f, Ag]=(f, Ag)_{\mathcal{C}}, \quad
\forall{f,g}\in\mathcal{D}(A).
$$
Hence, $A$ is a symmetric operator in $\mathfrak{H}_{\mathcal C}$.

\section{Schr\"{o}dinger operator with $\mathcal{PT}$-symmetric zero-range potentials}
\subsection{An example of $\mathcal{C}$-symmetry}
Let $\mathfrak{H}=L_2(\mathbb{R})$ and let $J=\mathcal{P}$, where
$\mathcal{P}f(x)=f(-x)$ is the space parity operator in
$L_2(\mathbb{R})$. In that case, the fundamental decomposition
\eqref{d1} of the Krein space $(L_2(\mathbb{R}), [\cdot, \cdot])$
takes the form
\begin{equation}\label{as5}
L_2(\mathbb{R})=L_2^{\mathrm{even}}{\oplus}L_2^{\mathrm{odd}},
\end{equation}
where $\mathfrak{H}_+=L_2^{\mathrm{even}}$ and
$\mathfrak{H}_-=L_2^{\mathrm{odd}}$ are subspaces of even and odd
functions in $L_2(\mathbb{R})$.

Consider the one-dimensional Schr\"{o}dinger operator with singular
zero-range potential
\begin{equation}\label{as3}
-\frac{d^2}{dx^2}+V_\gamma, \qquad
V_\gamma=i\gamma[<\delta',\cdot>\delta+
 <\delta,\cdot>\delta'], \quad \gamma\geq{0}
\end{equation}
where $\delta$ and $\delta'$ are, respectively, the Dirac
$\delta$-function and its derivative (with support at $0$).

It is easy to verify that
$\mathcal{PT}[-\frac{d^2}{dx^2}+V_\gamma]=[-\frac{d^2}{dx^2}+V_\gamma]\mathcal{PT}$,
where $\mathcal{T}$ is the complex conjugation operator
$\mathcal{T}f(x)=\overline{f(x)}$. Thus the expression \eqref{as3}
is $\mathcal{PT}$-symmetric \cite{B4, B22}.

The operator realization $A_\gamma$ of $-d^2/dx^2+V_\gamma$ in
$L_2(\mathbb{R})$ is defined as
\begin{equation}\label{new14}
A_\gamma=A_{\mathrm{reg}}\upharpoonright\mathcal{D}(A_\gamma), \quad
\mathcal{D}(A_\gamma)=\{f\in{{W_2^2}(\mathbb{R}\backslash\{0\})} \ :
\ A_{\mathrm{reg}}f\in{L_2(\mathbb{R})}\},
\end{equation}
where the regularization of $-d^2/dx^2+V_\gamma$ onto
${W_2^2}(\mathbb{R}\backslash\{0\})$ takes the form
$$
{A}_{\mathrm{reg}}=-\frac{d^2}{dx^2}+i\gamma[<\delta_{\mathrm{ex}}',\cdot>\delta+
<\delta_{\mathrm{ex}},\cdot>\delta'].
$$
Here $-{d^2}/{dx^2}$ acts on ${W_2^2}(\mathbb{R}\backslash\{0\})$ in
the distributional sense and
$$
 <\delta_{\mathrm{ex}}, f>=\frac{f(+0)+f(-0)}{2}, \quad <\delta_{\mathrm{ex}}', f>=-\frac{f'(+0)+f'(-0)}{2}
$$
for all $f(x)\in{{W_2^2}(\mathbb{R}\backslash\{0\})}$.

The operator $A_\gamma$ defined by \eqref{new14} is a
$\mathcal{P}$-self-adjoint operator in the Krein space
${L_2(\mathbb{R})}$.

According to \cite{AK}, $A_\gamma$ has ${\mathcal{C}}$-symmetry for
all $\gamma\not=2$. The corresponding operator
$\mathcal{C}\equiv\mathcal{C}_\gamma$ takes the form
\begin{equation}\label{as4}
\mathcal{C}_\gamma=\alpha_\gamma{\mathcal P}+i\beta_\gamma{\mathcal
R},
\end{equation}
where $\alpha_\gamma=\frac{\gamma^2+4}{|\gamma^2-4|}$ and
$\beta_\gamma=\frac{4\gamma}{|\gamma^2-4|}$ are `hyperbolic
coordinates' ($\alpha_\gamma^2-\beta_\gamma^2=1$) and ${\mathcal
R}f(x)=(\mathrm{sign} \ {x})f(x)$.

The operator $A_\gamma$ is reduced by the decomposition
\begin{equation}\label{as6}
 L_2(\mathbb{R})=\mathfrak{L}_+^\gamma[\dot{+}]\mathfrak{L}_-^\gamma, \qquad
 \mathfrak{L}_+^\gamma=(I+\mathcal{C}_\gamma)L_2(\mathbb{R}), \quad
 \mathfrak{L}_-^\gamma=(I-\mathcal{C}_\gamma)L_2(\mathbb{R}).
\end{equation}
Here $\mathcal{C}_\gamma={\mathcal P}(I-T_\gamma)(I+T_\gamma)^{-1}$,
where the operator $T_\gamma$
\begin{equation}\label{as7}
T_\gamma=i\frac{2}{\gamma}{\mathcal{RP}} \quad (\gamma>2);  \qquad
T_\gamma=i\frac{\gamma}{2}{\mathcal{RP}} \quad (\gamma<2)
\end{equation}
is the operator of transition from the fundamental decomposition
\eqref{as5} to \eqref{as6}.

Let us assume $\gamma>2$ (the case $\gamma<2$ is completely
similar). By \eqref{d4}, \eqref{as5}, and \eqref{as7}
$$
\mathfrak{L}_+^\gamma=\{f_\mathrm{even}+i\frac{2}{\gamma}{\mathcal{R}}f_\mathrm{even}
\ : \ f_\mathrm{even}{\in}L_2^{\mathrm{even}} \}, \quad
\mathfrak{L}_-^\gamma=\{f_\mathrm{odd}-i\frac{2}{\gamma}{\mathcal{R}}f_\mathrm{odd}
\ : \ f_\mathrm{odd}{\in}L_2^{\mathrm{odd}} \}.
$$

If $\gamma\to{2}$, then the invariant subspaces
$\mathfrak{L}_+^\gamma$ and $\mathfrak{L}_-^\gamma$ for $A_\gamma$
`tend' to each other and for $\gamma=2$ they coincide with the
hyper-maximal neutral subspace
$$
\mathfrak{L}^2=\{f_\mathrm{even}+i{\mathcal{R}}f_\mathrm{even} \ : \
f_\mathrm{even}{\in}L_2^{\mathrm{even}}\}=\{f_\mathrm{odd}-i{\mathcal{R}}f_\mathrm{odd}
\ : \ f_\mathrm{odd}{\in}L_2^{\mathrm{odd}}\}.
$$

The spectrum of $A_{2}$ coincides with $\mathbb{C}$ and any point
$z\in\mathbb{C}\setminus\mathbb{R}_+$ is an eigenvalue of $A_2$
\cite{AK}. Therefore, the $\mathcal{P}$-self-adjoint operator $A_2$
has neither $\mathcal{C}$-symmetry nor generalized
$\mathcal{C}$-symmetry (see Corollary \ref{p21}).

The obtained result is in accordance with the `physical' concept of
$\mathcal{PT}$-symmetry. Indeed, it is easy to verify that the
$\mathcal{PT}$-symmetry of \eqref{as3} is unbroken for
$\gamma\not=2$ and broken for $\gamma=2$ in the sense of \cite{B} --
\cite{BJ}. According to the general concepts of the theory
\cite{B22, BBJ}, the existence of a hidden $\mathcal{C}$-symmetry is
an intrinsic property of unbroken $\mathcal{PT}$-symmetry.

\subsection{An example of generalized $\mathcal{C}$-symmetry}
Let us consider a Hilbert space
$\mathfrak{H}=\bigoplus_{1}^{\infty}L_2(\mathbb{R})$ with elements
$\mathfrak{f}=\{f_1, f_2, \ldots\}$, where $f_i{\in}L_2(\mathbb{R})$
and the scalar product $(\cdot, \cdot)_{\mathfrak{H}}$ is defined by
the formula
 $$
 (\mathfrak{f},
 \mathfrak{g})_{\mathfrak{H}}=\sum_{i=1}^\infty(f_i,g_i)_{L_2(\mathbb{R})}.
 $$

 The operator $J\mathfrak{f}=\{{\mathcal{P}}f_1,
{\mathcal{P}}f_2, \ldots\}$ is a fundamental symmetry in ${\mathfrak
H}$ (i.e., $J=J^*$ and $J^2=I$) and $(\mathfrak{H}, [\cdot,\cdot])$
endowed by the indefinite metric $[\mathfrak{f},
 \mathfrak{g}]=(J\mathfrak{f},
 \mathfrak{g})_{\mathfrak{H}}$ is a Krein space.

The operator
$$
A_{\vec{\gamma}}\mathfrak{f}=\{A_{\gamma_1}f_1, A_{\gamma_2}f_2,
\ldots\}, \qquad \vec{\gamma}=\{\gamma_i\}, \quad \gamma_i\geq{0},
$$
where $A_{\gamma_i}$ are defined by \eqref{new14} is
$J$-self-adjoint in $\mathfrak{H}$. If $2$ is not a limit point for
the set $\vec{\gamma}$ (i.e., $2$ does not belong to the closure of
$\vec{\gamma}$), the operator $A_{\vec{\gamma}}$ has
$\mathcal{C}$-symmetry with
$\mathcal{C}=\bigoplus_{1}^{\infty}\mathcal{C}_{\gamma_i}$ where
$\mathcal{C}_{\gamma_i}$ are given by \eqref{as4}.

Let us assume that $\gamma_i\not=2$ ($i\in\mathbb{N}$) and there
exists a subsequence $\gamma_{j}$ of $\vec{\gamma}$ such that
$\gamma_j\to{2}$. In that case the operator
$T=\bigoplus_{1}^{\infty}T_{\gamma_i}$ with $T_{\gamma_i}$
determined by \eqref{as7} satisfies the conditions \eqref{e3} and
$\|T\|=1$. Therefore, $T$ is the operator of transition from the
fundamental decomposition of $(\mathfrak{H}, [\cdot,\cdot])$ to the
$J$-orthogonal sum \eqref{d2}. Since $\|T\|=1$ the subspaces
$\mathfrak{L}_{\pm}$ in \eqref{d2} are determined by the unbounded
operator
$$
\mathcal{C}=\bigoplus_{1}^{\infty}\mathcal{C}_{\gamma_i}=J(I-T)(I+T)^{-1}.
$$
Thus the operator $A_{\vec{\gamma}}$ has a generalized
$\mathcal{C}$-symmetry.

\section{Acknowledgments}
 The author thanks Dr. Uwe G\"{u}nther for valuable discussions. The support by
 DFG 436 UKR 113/88/0-1 research project is gratefully acknowledged.

 \end{document}